\documentclass[submission,copyright,creativecommons]{eptcs}
\providecommand{\event}{Theory of Quantum Computation, Communication and Cryptography (TQC) 2024} 

\usepackage{iftex}

\ifpdf
  \usepackage[T1]{fontenc}        
\else
  \usepackage{breakurl}           
\fi

\usepackage{todonotes}
\usepackage{tikz}

\usepackage{amsmath, amssymb, amsthm}
\usepackage{bbm}
\newcommand{\SWAP}{\operatorname{SWAP}}
\newtheorem{theorem}{Theorem}
\newtheorem{definition}{Definition}

\newtheorem{remark}{Remark}

\title{Symmetry-restricted quantum circuits are still well-behaved}
\author{Maximilian Balthasar Mansky
\institute{Institute of Informatics}
\institute{LMU Munich\\Munich, Germany}
\email{maximilian-balthasar.mansky@ifi.lmu.de}
\and
Santiago Londoño Castillo
\institute{Institute of Informatics}
\institute{LMU Munich\\Munich, Germany}
\and
Miguel Armayor-Martínez
\institute{Institute of Informatics}
\institute{LMU Munich\\Munich, Germany}
\and
Alejandro Bravo de la Serna
\institute{Institute of Informatics}
\institute{LMU Munich\\Munich, Germany}
\and
Gautham Sathish
\institute{Institute of Informatics}
\institute{LMU Munich\\Munich, Germany}
\and
Zhihao Wang
\institute{Institute of Informatics}
\institute{LMU Munich\\Munich, Germany}
\and
Sebastian Wölckert
\institute{Institute of Informatics}
\institute{LMU Munich\\Munich, Germany}
\and
Claudia Linnhoff-Popien
\institute{Institute of Informatics}
\institute{LMU Munich\\Munich, Germany}
}
\def\titlerunning{Symmetry-restricted quantum circuits are still well-behaved}
\def\authorrunning{Maximilian Balthasar Mansky et al}

\begin{document}
\maketitle

\begin{abstract}
We show that quantum circuits restricted by a symmetry inherit the properties of the whole special unitary group $SU(2^n)$, in particular composition, algebraic and topological closedness and connectedness. It extends prior work on symmetric states to the operators and shows that the operator space follows the same structure as the state space. The well-behavedness is independent of the symmetry requirement imposed on the subgroup. We provide an example of a permutation invariance across all qubits.

\end{abstract}

\section{Introduction}

Symmetries are a staple of modern physics. They are used to understand the underlying conservation laws, reduce the representation of a problem to a manageable size and generally represent physical structures \cite{sternberg_group_1995}. Beyond physics, they are an object of study in mathematics \cite{hall_lie_2013}. Thanks to their ubiquity, symmetries also become an object of study in the realm of quantum computing \cite{larocca_group-invariant_2022, meyer_exploiting_2023}. 

Restrictions on a larger space as a representation of a symmetry are generally interesting because they allow for a more intricate representation of the problem. Rather than exploring the whole solution space, understanding the underlying symmetry and integrating it into the solution space as a restriction reduces the available search space and should allow for faster convergence, especially in the case of parameterized applications \cite{schuld_introduction_2015}.

Symmetries are independent of the computation method. In the classical computing case, an appropriate choice for an algorithm allows for an efficient calculation. Graph isomorphisms are an example of symmetry and are exploited in graph neural networks \cite{bronstein_geometric_2021}. The general field of geometric deep learning captures symmetries and tries to express them through the architecture of neural networks.

In the context of machine learning, the most prominent example of symmetry exploitation is the convolutional neural network that integrates translational symmetry into image recognition \cite{wiatowski_mathematical_2018}. The effect is that the neural network can recognize an object in an image or data set independent of its position. This is in contrast with earlier approaches of dense neural networks that do not incorporate next neighbour relations and do not recognize inherent image symmetries. More recent transformer-based approaches \cite{vaswani_attention_2017} do not explicitly incorporate symmetries of the feature or solution space.

In the quantum computing domain, application-driven investigation has led to the implementation of symmetries, in general, \cite{larocca_group-invariant_2022, meyer_exploiting_2023} and specific \cite{song_geometry_2019, klus_symmetric_2021, mansky_permutation-invariant_2023}. The research generally takes inspiration from the classical domain with its much longer research history. Understanding the symmetry of the problem and incorporating it into the quantum circuits tailors the calculation to a smaller subspace. 


\section{Related work}

The structure Lie groups in general and the special unitary group $SU$ in particular can be understood via Lie group theory, for example through the standard works \cite{goodman_representations_2000, hall_lie_2013}. Most books only mention symmetry-restricted subspaces of $SU$ in passing and focus on complex of examples of Lie groups with special properties instead. Results specific to restricted (special) unitary groups are scattered across the literature and proven time and again for the specific use case of that work.

An excellent review article \cite{harrow_church_2013} collects results from the symmetry invariance of quantum states. Symmetric states, defined through a symmetric product of the state with all its permutations, have some interesting properties with regards to cloning and the de Finetti theorem \cite{konig_most_2009}. The view on the states as opposed on operators confers to the active vs passive view in quantum computing.

In contrast, symmetry-restricted quantum computing has been explored for its application to quantum machine learning \cite{larocca_group-invariant_2022, meyer_exploiting_2023}. Quantum circuits that respect the underlying symmetry of the problem are expected to converge faster to a solution. A side effect of the smaller solution space is that some parametrized quantum circuits with imposed symmetries do not exhibit barren plateaus \cite{ragone_unified_2023}. Barren plateaus is the observation that the local gradient for training loop vanishes, as a change in parameters does not result in a change in the loss function. The application-driven development is not situated on a solid theory, a gap that we intend to close with our work.

\section{Symmetry as a restriction}

We introduce the concept of symmetries on Lie groups, show its implementation for discrete groups via the special case of $\SWAP$ matrices and then generalize to any symmetry. In the subsections, we show that the desired properties of the special unitary group, in particular composition, closedness and connectedness, still hold for any imposed symmetry restriction.

A description of a physical system $U$ is invariant under a symmetry $S$ if it is left unchanged under the action of the symmetry,
\begin{equation}
    S U S^\dagger = U\label{eq:base-definition}
\end{equation}

This simple description extends to many different descriptions of physical systems. In the concrete case of quantum computing, we can represent $S$ and $U$ as $2^n \times 2^n$ matrices in $U(2^n)$ and $SU(2^n)$ respectively.\footnote{The two-qubit operations such as $\SWAP$ and $\operatorname{CNOT}$ are not part of $SU(2^n)$ due to negative eigenvalues. However, they are very useful for representing discrete symmetries.}

A discrete symmetry contains a countable number of operations that leave the underlying object invariant. The changes of the system under the symmetry operation are non-continuous. Discrete symmetries can be expressed through combinations of permutations – a property that we will exploit extensively here. This is easily seen in the effect on discrete groups itself, as shown in figure \ref{fig:simple-example}. The discrete symmetry operations on a square can be decomposed into permutation operations on the vertices of the square. However, not all single permutations are permitted. Consider the permutation $\pi_{1,2}$. it does not leave the square invariant and rather turns it into a bowtie.

\begin{figure}
    \centering
\begin{tikzpicture}[scale=1.7]
\draw (-.7, -.7) rectangle (.7,.7);
\draw (.7, -.7) node[fill=white] {\tiny 1} (.7, .7) node[fill=white] {\tiny 2} (-.7, .7) node[fill=white] {\tiny 3} (-.7, -.7) node[fill=white] {\tiny 4};
\draw[->] (-45:1.1) arc[radius=1.1, start angle=-45, end angle=45];
\draw[->] (-45:1.2) arc[radius=1.2, start angle=-45, end angle=135];
\draw[->] (-45:1.3) arc[radius=1.3, start angle=-45, end angle=225];

\draw[<->] (-.6, -.6) -- (-.6, .6);
\draw[<->] (.6, -.6) -- (.6, .6);
\draw[<->] (-.55, -.6) -- (.55, -.6);
\draw[<->] (-.55, .6) -- (.55, .6);

\draw[<->] (-.5, -.5) -- (.5, .5);
\draw[<->] (-.5, .5) -- (.5, -.5);

\begin{scope}[xshift=2cm, anchor=west]
\node at (0,.3) {Rotations $\pi_{1,2} \pi_{2,3} \pi_{3,4} \pi_{4,1}$};
\node at (0,0) {Edge mirrors $\pi_{1,2}\pi_{3,4}, \pi_{1,4} \pi_{2,3}$};
\node at (0,-.3) {Diagonals $\pi_{1,3}, \pi_{2,4}$};
\end{scope}
\end{tikzpicture}
    \caption{The square is invariant under rotations and mirroring across the diagonals and edges. The symmetry operations correspond to permutations of the indices of the vertices.}
    \label{fig:simple-example}
\end{figure}

Permutations have a natural representation in quantum computing through the $\SWAP$ operation. Its physical interpretation is the exchange of two information carrying objects, such as two qubits. Whether any physical movement happens at all depends on the underlying physical implementation. The matrix representation is given as:
\begin{equation}
    \SWAP=\begin{pmatrix}
        1 & 0 & 0 & 0\\
        0 & 0 & 1 & 0\\
        0 & 1 & 0 & 0\\
        0 & 0 & 0 & 1\\
    \end{pmatrix}
\end{equation}

For a discrete symmetry, the symmetry operation can then be composed of SWAPs between the relevant qubits. It is possible to identify the permutations $\pi$ of the discrete group symmetry with other two-qubit operations, such as CNOT. CNOT-invariant circuits behave the same way and inherit the same properties as shown here. We will explore the details of $\SWAP$-invariant circuits in detail in section \ref{sec:permutation-invariant}.

The $\SWAP$ matrix can itself be understood as a permutation matrix, in this case of the second and third entry. Any other permutation is also a valid symmetry, though it is more difficult to imagine a $\operatorname{CNOT}$ invariant circuit.

In fact, even continuous symmetries can be embedded into the symmetry $S$. Generally this requires an infinite set of matrices $S$ to represent a continuous symmetry, however the proofs below hold and the symmetry-respecting subspace inherits the properties of the superspace regardless. The exact realization of a discrete symmetry depends on the representation onto matrices and the representation of the problem onto the qubits.

\subsection{Composition}

\begin{theorem}[The symmetry restriction preserves the group properties]
For two elements $X_1, X_2$ in the symmetry-invariant space $\text{\emph{sym}}SU(2^n)\subset SU(2^n)$, the composition $X_2 \circ X_1$ is in $\text{\emph{sym}}SU$ as well. 
\end{theorem}

\begin{proof}
    Since $X_1,X_2$ are in the symmetry-invariant space $U$ one has $S \circ X_1 \circ S^\dagger = X_1$ and  $S \circ X_2 \circ S^\dagger = X_2$. We then want to show that the composition under the group operation is closed, i.e. $S \circ \left( X_2 \circ X_1 \right) \circ S^\dagger = X_2 \circ X_1 $.
    \begin{equation}
        S \circ \left(X_2 \circ  X_1 \right) \circ S^\dagger= S \circ X_2 \circ \left( S \circ S^\dagger \right) \circ X_1 \circ S^\dagger
    \end{equation}
    \begin{equation}
        = \left( S \circ X_2 \circ S^\dagger\right) \circ \left( S \circ X_1 \circ S^\dagger\right)= X_2 \circ X_1
    \end{equation} 
Thus, the group operation preserves the invariance.
\end{proof}
From this proof, we also inherit the group properties, existence of inverse, identity and associativity from the supergroup $SU(2^n)$.

The above proof holds in general for any group operation $\circ$. As a matrix group, the regular matrix multiplication is a natural choice for group operation in $U(2^n)$. Hitherto, it will assumed that the group operation is the matrix multiplication, i.e. $X \circ Y= XY$.\\

Moreover, note that imposing an invariance on a Lie group imposes the same symmetry on its corresponding Lie algebra by the linearization around the identity.  Below, we show that the Lie algebra $\mathfrak{u}$ corresponding to a symmetry-invariant space $U$ is algebraically closed under its Lie bracket.

\subsection{Algebraic Closedness}

We provide two proofs for the closedness of the subspaces, for the algebraic properties and the topological properties.

\begin{theorem}[Algebraic closedness]
    The algebra of a symmetric-restricted subspace $\text{\emph{sym}}\mathfrak{su}(2^n)$ of $\mathfrak{su}(2^n)$ is closed.
\end{theorem}
\begin{proof}
    Consider two elements $a, b \in \text{sym}\mathfrak{su}(2^n) $, such that $SaS^\dagger=a$ and $SbS^\dagger=b,\quad \forall S$. Then the Lie bracket
    \begin{align}
        [SaS^\dagger, SbS^\dagger] &= SaS^\dagger SbS^\dagger - SbS^\dagger SaS^\dagger \\
        &=SabS^\dagger - SbaS^\dagger\\
        &= S(ab - ba) S^\dagger\\
        &= S[a, b]S^\dagger
    \end{align}
    which shows that the commutator of two elements is still part of the symmetric subspace.
\end{proof}

\subsection{Topological Closedness}

\begin{theorem}
    A symmetry-restricted subspace of $SU(2^n)$, $\text{\emph{sym}}SU(2^n)$ is a closed manifold.
\end{theorem}

\begin{proof}

    Recall the definition of $\text{sym}SU(2^n):=\{A\in SU(2^n)\colon SAS^\dagger = A\}$. We will show that this set forms a closed submanifold of $\mathbb{R}^{2^n\times 2^n}$.
    
    To this end, consider the usual determinant map 
 \begin{equation}
 	\det\colon \mathcal{M}_{2^n\times 2^n}(\mathbb{R})\to \mathbb{R}
\end{equation}
 and the two maps $f_1$ and $f_S$ for all $S$ defined as
 \begin{align}
 f_1\colon\mathcal{M}_{2^n\times 2^n}(\mathbb{R})&\to\mathcal{M}_{2^n\times 2^n}(\mathbb{R})\nonumber\\
A &\to AA^\dagger\\
 f_S\colon \mathcal{M}_{2^n\times 2^n}(\mathbb{R})&\to\mathcal{M}_{2^n\times 2^n}(\mathbb{R})\nonumber\\
A &\to SA - AS
 \end{align}


    Then, we may view $\text{sym}SU(2^n)$ as
    
    \begin{equation}
        \text{sym}SU(2^n) = \bigcap_{S}f_S^{-1}(\{\textbf{0}\}) \;\bigcap\; f_1^{-1}(\{\mathbbm{1}\}) \;\bigcap\; \text{det}^{-1}(\{1\})
    \end{equation}

    and since the pre-image under a continuous map of a closed set is closed and the  intersection of closed sets is closed, the claim follows.
\end{proof}
\begin{remark}
    The proof holds for both finite intersections of closed sets (in the case of discrete symmetry restrictions) and infinite intersections (for continuous symmetries).
\end{remark}

\subsection{Connectedness}

Connectedness allows the complete parametrization of the space by guaranteeing that every point on the symmetry-invaraint space is connected to every other point. If a space is not connected, it decomposes into two or more disjoint spaces. 

We show connectedness via the centralizer of a the supergroup, $C_{U(2^n)}(H)$. The centralizer defines all symmetry elements that are invariant under the subgroup, essentially the inverse of the subgroup definition.

\begin{definition}[Group Centralizer]
Let $G$ be a group and $S$ be a subgroup of $G$. The centralizer of $S$ in $G$ is defined as
$$C_G(S):=\{g\in G|gs=sg,\forall s\in S\}.$$
\end{definition}
\begin{theorem}[The centralizer is connected]
Let $H$ be a subgroup of $U(2^n)$. Then the centralizer $C_{U(2^n)}(H)$ is connected.
\end{theorem}

\begin{proof}
For any $A\in C_{U(2^n)}(H)$ with eigenvalues $e^{i\theta_1},e^{i\theta_2},\dots,e^{i\theta_m}$ and multiplicities $k_1,k_2,\dots,k_m$ respectively, $A$ can be diagonalized by

\begin{equation}
A=PDP^{\dagger},\label{eq:apdp}
\end{equation}
with $P$ given by
\begin{equation*}
P=(P_1^{(1)}, \ldots, P_{k_1}^{(1)},P_1^{(2)},\ldots, P_{k_2}^{(2)},\dots,P_1^{(m)}, \ldots P_{k_m}^{(m)})
\end{equation*}
where each $P^{(i)}_j$ for $i\in \{1,\ldots, m\}$ and $j\in\{1,\ldots, k_i\}   $  denotes a (column) eigenvector and 
\begin{equation*}
D=\mbox{diag}(\underbrace{e^{i\theta_1},\dots,e^{i\theta_1}}_{k_1},\underbrace{e^{i\theta_2},\dots,e^{i\theta_2}}_{k_2},\dots,\underbrace{e^{i\theta_m},\dots,e^{i\theta_m}}_{k_m}).
\end{equation*}
Then for integers $1\leq i\leq m$, the eigenspace corresponding to $e^{i\theta_i}$ is 
\begin{equation}
    V_i:=\operatorname{span}\{P_1^{(i)}, \ldots, \;P_{k_i}^{(i)}\}
\end{equation}
We also have the following relation:
\begin{equation}
    AP^{(i)}_j = e^{i\theta_i}P^{(i)}_j,\quad\mbox{if } i \in\{1,\ldots, m\}, j\in \{1,\ldots, k_i\}
\end{equation}
Since $A\in C_{U(2^n)}(H)$, for each $S\in H\subset U(2^n)$, it holds that $A=SAS^{\dagger}=SPDP^{\dagger}S^{\dagger}=SPD(SP)^\dagger$. This means that $A$ can also be diagonalized by $SP$ for fixed eigenvalues and their positions. Since eigenvectors correspond to the same eigenvalue are in a same eigenspace, we have for $\forall S$, $i \in\{1,\ldots, m\}$ and $j\in \{1,\ldots, k_i\}$

\begin{equation}
\label{eq:SP_is_eigenvector}
SP^{(i)}_j\in V_i\quad \text{i.e.}\quad A(SP^{(i)}_j)=SDS^\dagger (SP^{(i)}_j)=e^{i\theta_i}(SP^{(i)}_j)
\end{equation}



We now define a curve $t\colon [0,1]\mapsto D(t)$ where
\begin{equation}
D(t)=\operatorname{diag}(\underbrace{e^{it\theta_1},\dots,e^{it\theta_1}}_{k_1},\underbrace{e^{it\theta_2},\dots,e^{it\theta_2}}_{k_2},\dots,\underbrace{e^{it\theta_m},\dots,e^{it\theta_m}}_{k_m})
\end{equation}
This verifies $D(0)=\mathbbm{1}$ and $D(1)=D$. We generalize the expression \eqref{eq:apdp} to the following parametrized version
\begin{equation}
A(t):=PD(t)P^\dagger
\end{equation}
which satisfies $A(0)=\mathbbm{1}$ and $A(1)=A$. Moreover, by construction, $A(t)$ verifies:


\begin{equation}
    A(t)P^{(i)}_j = e^{it\theta_i}P^{(i)}_j \quad i\in \{1,\ldots, m\}, j\in\{1,\ldots, k_i\}
\end{equation}

\noindent since when multiplying by $A=PDP^\dagger$, the matrix $P^\dagger$ transforms the vector $P^{(i)}_j$ appropriately to be multiplied by $D$. The only difference now is that we have $A(t)=PD(t)P^\dagger$, which only changes the value with which we multiply $P^{(i)}_j$, i.e. it is just a change of phase of the eigenvector.

Using this we want to show the invariance of $A(t)$ under the symmetry operation $S$, i.e. $SA(t)S^\dagger = A(t)$. First we note that for $i\in\{1,\ldots,m\}, j\in\{1,\ldots,k_i\}$, since $SP^{(i)}_j$ is an eigenvector of $A$ (see eq. \ref{eq:SP_is_eigenvector}), one can write

\begin{equation}
    SP^{(i)}_j = \sum_{l=1}^{k_i}c^{(i)}_lP^{(i)}_l\quad c^{(i)}_l\in \mathbb{R},\,\forall l\in\{1,\ldots, k_i\}
\end{equation}
and thus, for $i\in \{1,\ldots, m\}, j\in\{1,\ldots, k_i\}$


\begin{equation}
A(t)SP^{(i)}_j=\sum_{l=1}^{k_i}c^{(i)}_lA(t)P^{(i)}_l =\sum_{l=1}^{k_i}c^{(i)}_le^{it\theta_i}P^{(i)}_l = e^{it\theta_i}(SP^{(i)}_j)
\end{equation}


\noindent Therefore for $i\in\{1,\ldots,m\}$, the elements $\{SP^{(i)}_j\}_{j=1}^{k_i}$ are a group of eigenvectors of $A(t)$ corresponding to the eigenvalue $e^{it\theta_i}$ and thus $A(t)$ can be diagonalized not only by $P$, but also $SP$. Thus

\begin{equation}
SA(t)S^\dagger=S(PD(t)P^{\dagger})S^{\dagger}=(SP)D(t)(SP)^{\dagger}=A(t),\, \forall S\in H
\end{equation}
and one obtains that $A(t)\in C_{U(2^n)}(H)$ for all $t\in [0,1]$ and $A(t)$ is a path in $C_{U(2^n)}(H)$ connecting $\mathbbm{1}$ and $A$, which means that the centralizer $C_{U(2^n)}(H)$ is connected.
\end{proof}

The interpretation is that all symmetry invariant matrices $A$ are connected to each other via the identity. At minimum this leads to a star-like subspace structure, but since any combination of $A_i$ and $A_j$ is also part of the subspace, it is a continuous large region.

\begin{definition}
Let $H$ be a subgroup of $U(2^n)$. The (pseudo) centralizer of $H$ in $SU(2^n)$ is defined by
$$C_{SU(2^n)}(H):=\{A\in SU(2^n)|MA=AM, \forall M\in H\}.$$
\end{definition}
\begin{remark}
Usually $H$ is not a subgroup of $SU(2^n)$, but since $H$ and $SU(2^n)$ are both subgroups of $U(2^n)$ and preserve the operation on $U(2^n)$, $C_{SU(2^n)}(H)$ is well defined and it is easy to show that $C_{SU(2^n)}(H)$ is also a group.
\end{remark}

\begin{theorem}
$C_{SU(2^n)}(H)$ is connected.
\end{theorem}

\begin{proof}
The determinant
$\det:C_{U(2^n)}(H)\to U(1)$ is a group homomorphism and
\begin{equation}
    \ker(\det)=C_{SU(2^n)}(H)
\end{equation}
For any $e^{i\theta}\in U(1)$, we have $e^{i\frac{\theta}{2^n}}\mathbbm{1}\in C_{U(2^n)}(H)$. Hence
\begin{equation}
    \mbox{im}(\det)=U(1)
\end{equation}
By groups isomorphism theorem we have
\begin{equation}
    C_{U(2^n)}(H)/C_{SU(2^n)}(H)\simeq U(1)
\end{equation}
$C_{U(2^n)}(H)$ is connected so $C_{SU(2^n)}(H)$ is also connected.
\end{proof}

From the connectedness of the centralizer it follows that the symmetry-restricted subgroup is connected. It itself acts as a centralizer with respect to the symmetry.


\section{Example – permutation invariance}\label{sec:permutation-invariant}

The strongest discrete symmetry is the permutation invariance, where any input is interchangeable with any other input. Any other discrete symmetry can be realized as a subset of permutaitons. In the quantum computing setting, it can be realized by requiring all qubits to be interchangeable. An exchange of individual qubits $i$ and $j$ can be realized by the $\SWAP_{i,j}$ operation. The symmetry action on the system is then defined by the power set of all possible $\SWAP$ operations, $S=\{\SWAP_{i,j}\}$.\cite{mansky_permutation-invariant_2023} We will refer to this space as $\text{pi}SU(2^n)$, for \emph{p}ermutation \emph{i}nvariant special unitary group under all $\SWAP$s. 

\begin{equation}
    \text{pi}SU(2^n) = \left\{U| SUS = U, \quad\forall S \in \{\SWAP_{i,j}\}_{i,j=1}^n\right\}
\end{equation}

The space of permutation invariant quantum circuits can be created constructively via the corresponding Lie algebra $\text{pi}\mathfrak{su}(2^n)$. Their elements follow the same restriction as \eqref{eq:base-definition}. Via the properties of the $\SWAP$, one can also express the elements via the Pauli group, 
\begin{align}
    \text{pi}\mathfrak{su}(2^n) &= \{x | SxS = x, \quad\forall S \in \{\SWAP_{i,j}\}\} \nonumber\\
    &=\left\{ \bigotimes_{i=1}^n | \pi_i \sigma_1 \otimes \sigma_2 \otimes \ldots \otimes \sigma_n = \pi_j \sigma_1 \otimes \sigma_2 \otimes \ldots \otimes \sigma_n \forall \pi_i, \pi_j\right\}
\end{align}
where $\pi$ denotes a permutation of the Pauli string and $\sigma$ is one of $\{\sigma_x, \sigma_y, \sigma_z, \mathbbm{1}\}$. A permutation $\pi$ of the Pauli string corresponds naturally to symmetry operations $S$ in the $\{\SWAP_{i,j}\}$ group \cite{nielsen_quantum_2010}. Using the Pauli string formalism, it is possible to create all elements of $\text{pi}\mathfrak{su}$ by realizing that a sum over all permutations of a particular Pauli string is permutation invariant,
\begin{equation}
    \text{pi}\mathfrak{su}(2^n) = \operatorname{span}\left\{\sum_k \pi_k \bigotimes_{i=1}^n\sigma_i\right\} \backslash \mathbbm{1}^n
\end{equation}

This equation creates sums of strings with all possible permutations of the indices. Since a $\SWAP$ exchanges the same two indices $i, j$ on all elements of the sum, the overall effect is just a reordering of elements in the sum. This allows us to create a commuting diagram for $\text{pi}SU$, shown in figure \ref{fig:pisu-diagram}.

\begin{figure}
    \centering
    \begin{tikzpicture}[xscale = 4, yscale=-2]
        \draw (0,0) node (su) {$\mathfrak{su}$}
            (0, 1) node (SU) {$SU$}
            (1, 0) node (pisu) {$\text{pi}\mathfrak{su}$}
            (1, 1) node (piSU) {$\text{pi}SU$};
        \draw[->] (su) -- (SU) node[midway, left] {$\exp$};
        \draw[->] (su) -- (pisu) node[midway, above] {$\SWAP$ invariance} node[midway, below] {symmetrization};
        \draw[->] (pisu) -- (piSU) node[midway, right] {$\exp$};
        \draw[->] (SU) -- (piSU) node[midway, below] {$\SWAP$ invariance};
    \end{tikzpicture}
    \caption{The commuting diagram corresponding to $\text{pi}SU$. The restriction on the algebra creates $\text{pi}\mathfrak{su}$, which generates $\text{pi}SU$ via the exponential function, and all elements of the group $SU$ can be restricted to be $\SWAP$ invariant to create $\text{pi}SU$.}
    \label{fig:pisu-diagram}
\end{figure}
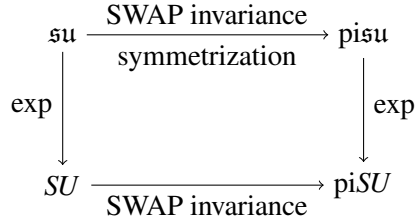

Because the elements of the group can also be epressed as generators in the algebra, it possible to find circuits directly for the elements. This can be done by Pauli string exponentiation \cite{mansky_decomposition_2023}. The resulting circuits have the structure of a $\operatorname{CNOT}$ ladder with rotations around them to shift them to the relevant axis. The Pauli strings can be exponentiated individually to build circuits and then concatenated, as long as they only contain two of the three Pauli matrices \cite{mansky_permutation-invariant_2023}. 

\begin{equation}
    \exp\left(-\frac{i\alpha}2\sum_k \pi_k \bigotimes_{i=1}^n \sigma_i\right) = \prod_{k} \exp\left(-\frac{i\alpha}2\pi_k \bigotimes_{i=1}^n \sigma_i\right)\quad\text{ if } \sigma= \{\sigma_x, \sigma_y, \mathbbm{1}\}, \{\sigma_x, \sigma_z, \mathbbm{1}\} \text{ or } \{\sigma_y, \sigma_z, \mathbbm{1}\}
\end{equation}

Therefore the creation of symmetry-respecting circuits, at least those constructed via a $\SWAP$-based symmetry representation, is straightforward. Each Pauli string exponential has a direct representation on the quantum circuit as a parametrized element with parameter $\alpha$ \cite{mansky_decomposition_2023}.

\section{Discussion}

It is often assumed implicitly that the subspaces spanned by symmetry restrictions have all the necessary properties for quantum computing. We have shown that this assumption is warranted and the subspace inherits the properties of the superspace $SU(2^n)$. With the constructive method of taking symmetry restrictions as permutation matrices, we have shown that the symmetry-restricted subspace is always closed and connected. Based on our results, it allows the use of symmetric subspaces at ease within quantum computing applications.

We also expand prior work on symmetric states to the operator formalism. We show that the quantum circuits permit the same symmetry considerations as quantum states. This also holds true for continuous symmetries, whereas the symmetric states rely on a discrete symmetry \cite{harrow_church_2013}. However, it is straightforward to extend the symmetric product for the states to represent continuous symmetries as well.

The proofs presented here hold for subgroups of the special unitary group and unitary group. The results presented here are also not easily transferable to classical machine learning, since the underlying Euclidian vector space for the feature and label space has a very different structure to $SU$ and $U$.

\section*{Acknowledgements}

The authors acknowledge funding by the German Bundesministerium für Bildung und Forschung (BMBF) under grant 13N16089 (BAIQO) as part of the funding program "Förderprogramm Quantentechnologien – von den Grundlagen zum Markt" (funding program quantum technologies – from the basics to the market).

\bibliographystyle{eptcs}
\bibliography{references}
\end{document}